\RequirePackage{fix-cm}
\documentclass[smallextended,numbook]{svjour3}       % onecolumn (second format)
\smartqed  % flush right qed marks, e.g. at end of proof
\usepackage{graphicx}
\usepackage{amsfonts}
\usepackage[english]{babel}
\usepackage{color}
\usepackage{booktabs}
\usepackage[english]{babel}
\usepackage{cite}
\usepackage{algorithm, algorithmic}
\usepackage{setspace}
\usepackage{braket}
\usepackage[justification=centering]{caption}
\newtheorem{algo}{Algorithm}[section]

\begin{document}

\title{Iterative methods for computing U-eigenvalues of non-symmetric complex tensors with application in  quantum entanglement%\thanks{Grants or other notes
%about the article that should go on the front page should be
%placed here. General acknowledgments should be placed at the end of the article.}
}
%\subtitle{Do you have a subtitle?\\ If so, write it here}

\titlerunning{Iterative methods for computing U-eigenvalues of non-symmetric complex tensors}   % if too long for running head

\author{Mengshi Zhang\and Guyan Ni    \and
        Guofeng Zhang %etc.
}

%\authorrunning{Short form of author list} % if too long for running head

\institute{Mengshi Zhang \at
              Department of Mathematics, National University of
              Defense Technology, Changsha, Hunan 410073, China.  \\
              \email{msh$_{-}$zhang@163.com}
              \and
              Guyan Ni \at
              Corresponding Author. Department of Mathematics, National University of
              Defense Technology, Changsha, Hunan 410073, China.  \\
              \email{guyan-ni@163.com}           %  \\
%             \emph{Present address:} of F. Author  %  if needed
           \and
           Guofeng Zhang \at
             Department of Applied Mathematics,
The Hong Kong Polytechnic University, Hong Kong, 999077, China.\\
              \email{guofeng.zhang@polyu.edu.hk}
}

\date{Received: date / Accepted: date}
% The correct dates will be entered by the editor

\maketitle

\begin{abstract}
The purpose of this paper is to study the problem of computing unitary eigenvalues (U-eigenvalues) of non-symmetric complex tensors.  By means of  symmetric embedding of complex tensors,  the relationship between U-eigenpairs of a non-symmetric complex tensor and unitary symmetric eigenpairs (US-eigenpairs) of its symmetric embedding tensor is established.  An algorithm  (Algorithm \ref{algo:1}) is given to compute the U-eigenvalues of non-symmetric complex tensors by means of symmetric embedding. Another algorithm, Algorithm \ref{algo:2}, is proposed to directly compute the U-eigenvalues of non-symmetric complex tensors, without the aid of symmetric embedding.  Finally, a tensor version of the well-known Gauss-Seidel method is developed. Efficiency of these three algorithms are compared by means of various numerical examples. These algorithms are applied to compute the geometric measure of entanglement of quantum multipartite non-symmetric pure states.
% \PACS{PACS code1 \and PACS code2 \and more}
\end{abstract}

\section{Introduction}

There is a variety of ways to define tensor eigenvalues, e.g., Z-eigenvalue\cite{qi05}, H-eigenvalue\cite{L05}, U-eigenvalue\cite{NQB14} and generalized eigenvalue\cite{KKT08}. The problem of computing the eigenvalues of a tensor has been proved to be NP-hard\cite{Hill13}.  An increasing number of numerical methods have been proposed in the last decades.   Kolda et al. \cite{kdm11} introduced shifted symmetric high order power method (SS-HOPM) to compute Z-eigenvalues of symmetric real tensors.  By means of a Jacobian semidefinite programming (SDP) relaxation method \cite{Las2001},   Z-eigenvalues of real tensors have been computed in Nie et al. \cite{NieW2014}  and Cui et al. \cite{CuiDN2014}. There are also several methods for computing Z-eigenvalues of real tensors, see, e.g., \cite{QWW09,HCD15}.  Kolda et al. \cite{kdm14} extended SS-HOPM to an adaptive shifted symmetric high order power method for computing generalized tensor eigenvalues. Chen et al. \cite{CHL2016} studied the generalized tensor eigenvalue problem  via homotopy methods. An adaptive gradient method for computing generalized tensor eigenpairs has been developed in \cite{YYXSZ16}. Fu et al. \cite{FJL18} derived new algorithms to compute best rank-one approximation of conjugate partial-symmetric (CPS) tensors by unfolding CPS tensors to Hermitian matrices. Che et al. \cite{CCW17} proposed a neural networks method for computing the generalized eigenvalues of complex tensors, and Che et al.\cite{CQW17} also derived an iterative algorithm for computing US-eigenpairs of complex symmetric tensors and U-eigenpairs of complex tensors based on the Takagi factorization of complex symmetric matrices. Hua et al. \cite{HNZ16} computed the largest US-eigenvalue of a symmetric complex tensor.  Ni et al. \cite{NB16} computed US-eigenpairs for symmetric complex tensors via the spherical optimization with complex variables.  However, there are relatively few studies on computing eigenvalues of non-symmetric complex tensors.

This paper aims to propose three methods to compute U-eigenvalues of non-symmetric complex tensors. We first build a one-to-one correspondence between a U-eigenpair of a non-symmetric complex tensor and a US-eigenpair of its symmetric embedding. Based on symmetric embedding,  Algorithm \ref{algo:1} is proposed to compute U-eigenpairs of non-symmetric complex tensors. Unfortunately, due to symmetric embedding, the size of the resulting tensor used in Algorithm \ref{algo:1} is usually very large, which significantly affects the  computational efficiency  of Algorithm \ref{algo:1}. To circumvent this difficulty,   Algorithm \ref{algo:2} is proposed to   compute the  U-eigenpairs of non-symmetric complex tensors directly. Convergence of Algorithms \ref{algo:1} and \ref{algo:2} are  established.  Finally, Algorithm \ref{algo:3}, a tensor version of the Gauss-Seidel method, is proposed.

Quantum entanglement was first introduced by Einstein and Schr\"{o}dinger\cite{Ein35,Sch35}, and it is regarded as one of the most important and fundamental notions in quantum information\cite{MI2000}. The geometric measure of entanglement (GME) is one of the most important and widely used measures for quantum entanglement\cite{MI2000,DLJ2001,FDO2016,QL2017,RT2014,Ben96,Ve97,Har03}. The GME was first proposed by Shimony \cite{S95} in 1995 for bipartite states and generalized to multipartite states by Wei and Goldbart \cite{WG03} in 2003.  Mathematically, a quantum pure state can be described in terms of a tensor (or hypermatrix), thus the problem of computing the GME of a pure state can be converted into a tensor eigenvalue computation problem\cite{NQB14,Hill10,hqz12,QZN2017,NZZ2017,HNZ16,Hay09}.  Theorem 1 in \cite{hqz12} indicates that the GME of a symmetric pure state with non-negative probability amplitudes is equal to the largest Z-eigenvalue of the corresponding non-negative tensor. Ni et al.\cite{NQB14}  found that for some real tensors, the  largest absolute-value of the Z-eigenvalue may not be equal to the entanglement eigenvalue. Hence, Ni et al.\cite{NQB14} introduced the concept of U-eigenvalue of complex tensors, and showed that the problem of computing the entanglement eigenvalue of a pure state is equivalent to the problem of computing the largest U-eigenvalue of the corresponding tensor. As an application, we apply  Algorithms \ref{algo:1}-\ref{algo:3} to compute the GME of quantum pure states.

The paper is organized as follows. In Section 2, we introduce the concept of GME of multipartite pure states,   U-eigenpairs of tensors, and tensor blocking. In Section 3, we propose three algorithms to compute the U-eigenvalues of non-symmetric complex tensors, and some basic theorems are also proved. In Section 4, we present numerical examples for various non-symmetric pure states, and compare the efficiency of these three algorithms. Section 5 concludes this paper.

%-------- new section -----------------------------------------------
\section{Preliminaries}
This section introduce  the geometric measure of  entanglement of quantum multipartite pure states, U-eigenvalues of complex tensors, and tensor blocking. Interested reader may refer to\cite{NQB14, hqz12, NZZ2017, RV2013} for details.

\subsection{Geometric measure of entanglement (GME) of multipartite pure states}
Quantum states are fundamental quantities in a quantum system. For an $m$-partite quantum system where the dimension of  the $k$th party is $n_k$, $(k=1, \ldots, m)$, its pure states are elements of the tensor product space $H=\otimes^m_{k=1}\mathbb{C}^{n_k} \equiv \mathbb{C}^{n_1\times\cdots \times n_m}$. Let $\{|e^{(k)}_{i_k}\rangle : i_k=1,2,...,n_k\}$ be an orthonormal basis of
$\mathbb{C}^{n_k}$. Then $\{|e^{(1)}_{i_1}e^{(2)}_{i_2}\cdots e^{(m)}_{i_m}\rangle : i_k=1,2,...,n_k; k=1,2,...,m\}$ is an orthonormal basis of $H$. A pure state $|\psi\rangle \in H$ can be written as
\begin{equation}\label{state1}
|\psi\rangle :=\sum^{n_1,...,n_m}_{i_1,...,i_m=1}\mathcal{X}_{i_1...i_m}|e^{(1)}_{i_1}\cdots e^{(m)}_{i_m}\rangle,
\end{equation}
where $\mathcal{X}_{i_1...i_m}\in\mathbb{C}$. $|\psi\rangle$ is called symmetric if $\mathcal{X}_{i_1...i_m}$ remains the same for all permutations of indices  $\{i_1,...,i_m\}$. Let
\begin{equation}\label{state2}
|\varphi\rangle :=\sum^{n_1,...,n_m}_{i_1,...,i_m=1}\mathcal{Y}_{i_1...i_m}|e^{(1)}_{i_1}\cdots e^{(m)}_{i_m}\rangle.
\end{equation}
be another pure state. The inner product and norm of multipartite pure states are denoted as
$$
\langle\psi|\varphi\rangle := \sum^{n_1,...,n_m}_{i_1,...,i_m=1}\mathcal{X}^*_{i_1...i_m}\mathcal{Y}_{i_1...i_m}, ~~|||\varphi\rangle||:=\sqrt{\langle\varphi|\varphi\rangle},
$$
where $\mathcal{X}^*_{i_1...i_m}$ is the complex conjugate of $\mathcal{X}_{i_1...i_m}$.  If $|||\varphi\rangle||=1$, then the state $|\varphi\rangle$ is a normalized state, also called unit state. All the quantum states used in this paper are assumed to be normalized states.
\begin{definition}
  An $m$-partite pure state $|\phi\rangle$ is called separable, if it can be written as
$$
|\phi\rangle:=\otimes^m_{k=1}|\phi^{(k)}\rangle,
$$
where $|\phi^{(k)}\rangle\in \mathbb{C}^{n_k}$. If an $m$-partite pure state is not separable, then it is an entangled state.
\end{definition}

The set of all separable normalized pure states in $H$  is denoted as $Separ(H)$. The  GME of a given multipartite pure state $|\psi\rangle$ can be defined as the distance between $|\psi\rangle$ and $Separ(H)$
\begin{equation}\label{distance}
  E_G(|\psi\rangle):=\min_{|\phi\rangle\in Separ(H)}|||\psi\rangle-|\phi\rangle||.
\end{equation}
As the objective function in (\ref{distance}) is a continuous function on a compact set in a finite dimensional space, the minimizer of (\ref{distance}) does exist. Actually, (\ref{distance}) can be converted to a maximization problem $$G(|\psi\rangle)=\max_{|\phi\rangle\in Separ(H)} |\langle\psi|\phi\rangle|.$$
$G(|\psi\rangle)$ is called the maximal overlap of a given $m$-partite pure state, and it is also called
 entanglement eigenvalue in \cite{WG03}.
By \cite{QZN2017}, the GME of  $|\psi\rangle$ is equal to
\begin{equation}\label{geometric measure1}
E_G(|\psi\rangle)=\sqrt{2-2G(|\psi\rangle)}.
\end{equation}
Clearly, the smaller the maximal overlap is, the larger the GME of $|\psi\rangle$ is. In quantum physics, a large geometric measure usually  indicates that the state $|\psi\rangle$ is more entangled.
\subsection{Complex tensors and their U-eigenpairs}
For the pure state $|\psi\rangle$ defined in (\ref{state1}), one can define an $m$th-order complex tensor  $\mathcal{X}=(\mathcal{X}_{i_1\cdots i_m})$. In other words, each quantum pure state corresponds to a complex tensor. Hence, we can calculate the GME of a quantum state by means of its corresponding complex tensor. If $\mathcal{X}_{i_1...i_m}$ remains the same for all permutations of indices  $\{i_1,...,i_m\}$, then $\mathcal{X}$ is called symmetric. For $\mathcal{X},\mathcal{Y}\in H$, the inner product and norm are defined as
  $$
  \langle \mathcal{X},\mathcal{Y}\rangle:=\sum^{n_1,...,n_m}_{i_1,...,i_m=1} \mathcal{X}^*_{i_1...i_m}\mathcal{Y}_{i_1...i_m}, \ \ ||\mathcal{X}||:=\sqrt{\langle \mathcal{X},\mathcal{X}\rangle},
  $$
where $\mathcal{X}^*_{i_1...i_m}$ denotes the complex conjugate of $\mathcal{X}_{i_1...i_m}$.

A tensor can be geometrically viewed as a multi-liner function, and it can be represented by a linear combination of outer products of vectors. Let $\mathbf{x}^{(i)}\in \mathbb{C}^{n_i}$, $i=1,\cdots,m$, their outer product $\mathbf{x}^{(1)}\otimes \mathbf{x}^{(2)}\otimes \cdots\otimes \mathbf{x}^{(m)}$, denoted as $\otimes_{i=1}^m \mathbf{x}^{(i)}$, is called a rank-one tensor, whose components are
$$(\otimes_{i=1}^m \mathbf{x}^{(i)})_{i_1...i_m}:=\mathbf{x}^{(1)}_{i_1}\cdots \mathbf{x}^{(m)}_{i_m}.$$
In particular, when $\mathbf{x}^{(1)}=\cdots=\mathbf{x}^{(m)}=\mathbf{x}\in\mathbb{C}^{n}$, we write $\otimes_{i=1}^m \mathbf{x}^{(i)}$ as $\otimes_{i=1}^m \mathbf{x}$ or simply as $\mathbf{x}^m$, which clearly is a symmetric rank-one complex tensor.

By the notation in \cite{NQB14}, for $\mathcal{T}\in H$, we denote the inner product between $\mathcal{T}$ and a rank-one tensor $\otimes_{i=1}^m \mathbf{x}^{(i)}$ by a homogenous polynomial
$$
\langle \mathcal{T},\otimes_{i=1}^m \mathbf{x}^{(i)}\rangle\equiv \mathcal{T}^{*}\mathbf{x}^{(1)}\cdots \mathbf{x}^{(m)} :=\sum_{i_1,\cdots,i_m=1}^{n_1,\cdots,n_m} \mathcal{T}_{i_1\cdots i_m}^* x_{i_1}^{(1)}\cdots x_{i_m}^{(m)}.
$$
Moreover, $\langle \mathcal{T},\otimes_{i=1,i\not=k}^m \mathbf{x}^{(i)}\rangle$ denotes a vector in $\mathbb{C}^{n_k}$, whose $i_k$-th components are
$$
\langle \mathcal{T},\otimes_{i=1,i\not=k}^m \mathbf{x}^{(i)}\rangle_{i_k} :=\sum_{i_1,\cdots,i_{k-1},i_{k+1},\cdots,i_m=1}^{n_1,\cdots,n_{k-1},n_{k+1},\cdots,n_m} \mathcal{T}_{i_1\cdots i_k\cdots  i_m}^* x_{i_1}^{(1)}\cdots x_{i_{k-1}}^{(k-1)}x_{i_{k+1}}^{(k+1)} \cdots x_{i_m}^{(m)}.
$$
Finally, we define a new complex vector $\langle \otimes_{i=1,i\not=k}^m \mathbf{x}^{(i)},\mathcal{T}\rangle:=\langle \mathcal{T},\otimes_{i=1,i\not=k}^m \mathbf{x}^{(i)}\rangle^* $.
The definition of U-eigenpair of a complex tensor $\mathcal{T}$, introduced in \cite{NQB14}, is   given below.

\begin{definition}(\cite{NQB14})
For an $m$th-order tensor $\mathcal{T}\in H$, a tuple \\$\{\lambda,(\mathbf{x}^{(1)}, \cdots, \mathbf{x}^{(m)})\}$ with $\lambda\in\mathbb{R},  \mathbf{x}^{(i)}\in \mathbb{C}^{n_i},  i=1,\cdots,m$ is called a U-eigenpair of $\mathcal{T}$ if $\lambda$ and the rank-one tensor $\otimes_{i=1}^m \mathbf{x}^{(i)}$ are solutions of the following system of equations:
\begin{equation}\label{EQ:Ueigen}
\left\{
  \begin{array}{ll}
    \langle \mathcal{T},\otimes_{i=1,i\not=k}^m \mathbf{x}^{(i)}\rangle=\lambda {\mathbf{x}^{(k)}}^*, \\
    \langle \otimes_{i=1,i\not=k}^m \mathbf{x}^{(i)},\mathcal{T}\rangle=\lambda {\mathbf{x}^{(k)}}, \\
    \lambda\in \mathbb{R}, ||\mathbf{x}^{(i)}||=1,i=1,2,\cdots, m,
  \end{array}
\right.\ k=1,2,\cdots, m.
\end{equation}
\end{definition}
 In fact, (\ref{EQ:Ueigen}) is equivalent to
\begin{equation}\label{Ueigen1}
\left\{
  \begin{array}{ll}
    \langle \mathcal{T},\otimes_{i=1,i\not=k}^m \mathbf{x}^{(i)}\rangle=\lambda {\mathbf{x}^{(k)}}^*, \\
    \lambda\in \mathbb{R}, ||\mathbf{x}^{(i)}||=1,\ i=1,2,\cdots,m,
  \end{array}
\right.\ k=1,2,\cdots, m,
\end{equation}
or
\begin{equation}\label{Ueigen}
\left\{
  \begin{array}{ll}
    \langle \otimes_{i=1,i\not=k}^m \mathbf{x}^{(i)},\mathcal{T}\rangle=\lambda {\mathbf{x}^{(k)}},\\
    \lambda\in \mathbb{R}, ||\mathbf{x}^{(i)}||=1,\ i=1,2,\cdots,m,
  \end{array}
\right.\ k=1,2,\cdots, m.
\end{equation}

When $\mathcal{S}$ is a symmetric tensor, we call the tuple $\{\lambda,\mathbf{x}\}$ a US-eigenpair of $\mathcal{S}$ if the scalar $\lambda$ and the vector $\mathbf{x}$ satisfy
\begin{equation}\label{EQ:USeigen0}
\left\{
  \begin{array}{ll}
    \langle \mathcal{S},\otimes_{i=1}^{m-1} \mathbf{x}\rangle=\lambda \mathbf{x}^*, \\
    \lambda\in \mathbb{R}, \mathbf{x}\in \mathbb{C}^{n}, ||\mathbf{x}||=1.
  \end{array}
\right.
\end{equation}
or
\begin{equation}\label{EQ:USeigen0}
\left\{
  \begin{array}{ll}
    \langle \otimes_{i=1}^{m-1} \mathbf{x},\mathcal{S}\rangle=\lambda \mathbf{x}, \\
    \lambda\in \mathbb{R}, \mathbf{x}\in \mathbb{C}^{n}, ||\mathbf{x}||=1.
  \end{array}
\right.
\end{equation}
The relationships between the US-eigenvalue and other definitions of tensor eigenvalues are discussed in \cite{JLZ16}.

Given a tensor $\mathcal{T}\in H$, A rank-one complex tensor $\otimes_{i=1}^m \mathbf{x}^{(i)}$ is called the best complex rank-one approximation to $\mathcal{T}$ if it is the minimizer of the optimization problem
 \begin{equation}\label{best appro1}
 \min_{\mathbf{x}^{(i)}\in\mathbb{C}^{n_i},||\mathbf{x}^{(i)}||=1}||\mathcal{T}-\otimes^m_{i=1} \mathbf{x}^{(i)}||.
 \end{equation}
It is proved in \cite{NQB14,HKWGG2009} that for a symmetric $m$th-order complex tensor $\mathcal{S}$, its best symmetric complex rank-one approximation is also the best complex rank-one approximation, in other words, for a symmetric complex tensor   $\mathcal{S}$, the optimization problem (\ref{best appro1}) reduces to the following one
 \begin{equation}\label{best appro2}
  \min_{\mathbf{x}\in\mathbb{C}^n,||\mathbf{x}||=1}||\mathcal{S}-\otimes^m_{i=1} \mathbf{x}||.
 \end{equation}
As pointed out by \cite{NZZ2017}, (\ref{best appro1}) is equivalent to the following maximization problem
\begin{equation}\label{max appro1}
\left\{
  \begin{array}{ll}
   \max ~~|\langle \mathcal{T},\otimes^m_{i=1} \mathbf{x}^{(i)}\rangle| \\
   s.t.~~||\mathbf{x}^{(i)}||=1, ~~\mathbf{x}^{(i)}\in\mathbb{C}^{n_i}, i=1,\ldots, m.
  \end{array}
\right.
\end{equation}
Moreover, let $ (\mathbf{z}^{(1)},\cdots, \mathbf{z}^{(m)})$ be a solution to (\ref{max appro1}).  It is shown in \cite{NZZ2017} that the largest U-eigenvalue of the tensor $\mathcal{T}$ is actually $\max |\langle \mathcal{T},\otimes^m_{i=1} \mathbf{z}^{(i)}\rangle|$ and  $ (\mathbf{z}^{(1)},\cdots, \mathbf{z}^{(m)})$ are the corresponding U-eigenvectors.  Also, the rank-1 tensor $ \otimes_{i=1}^m \mathbf{z}^{(i)}$ is the best complex rank-one approximation of $\mathcal{T}$.

The following result has been proved in \cite{NZZ2017}.
\begin{theorem}\label{TH:Glambdamax} Assume that $\mathcal{T}$ is the corresponding tensor of a multipartite pure state $|\psi\rangle$ under an orthonormal basis as in (\ref{state1}). Let $\lambda_{max}$ be the largest U-eigenvalue of $\mathcal{T}$. Then

(1) $G(|\psi\rangle)=\lambda_{max}$,

(2) $ E_G(|\psi\rangle)=\sqrt{2-2\lambda_{max}}.$
\end{theorem}

This theorem makes it possible to investigate the  GME of a multipartite pure state by means of  the U-eigenpairs of the associated complex tensor.

\subsection{Tensor blocking}
We introduce how to block a complex tensor, in analogy to the real tensor case proposed in \cite{RV2013}. Let $\mathcal{T}\in\mathbb{C}^{n_1\times\cdot\cdot\cdot\times n_m}$ be an $m$th-order tensor. If $a$ and $b$ are integers with $a\leq b$, then let $a:b$ denote the row vector $[a, a+1, \cdots, b]$. Blocking the tensor $\mathcal{T}$ is the act of partitioning its index range vectors $1:n_1$, $1:n_2$, $\cdots$, $1:n_m$ in the following way. For each $k =1,\cdots, m$, let $\mathbf{r}^{(k)} \equiv 1:n_k $.  Partition $\mathbf{r}^{(k)}$ into $b_k$ blocks as
\begin{equation}\label{Seg}
\mathbf{r}^{(k)} = [\mathbf{r}_1^{(k)}, \cdots, \mathbf{r}_{b_k}^{(k)}].
\end{equation}
Let $\rho_{i}^{(k)}=\sum_{j=1}^{i-1} \mathrm{length}(\mathbf{r}_{j}^{(k)})$ for each $i=1,\cdots,b_k$.
Then $\mathcal{T}$ can be regarded as a $b_1\times\cdot\cdot\cdot\times b_m$ block tensor. Let $\mathbf{i}=\{i_1, \cdots, i_m\}$ where $1\leq i_k \leq b_k$. The $\mathbf{i}$-th block is denoted as
$$\mathcal{T}_{[\mathbf{i}]}=\mathcal{T}_{[i_1, \cdots, i_m]}.$$
To be specific, let $\mathbf{j}=\{j_1,\cdots,j_m\}$, where $  j_k=1,\cdots,\mathrm{length}(\mathbf{r}_{i_k}^{(k)})$. Then the $\mathbf{j}$-th entry of the subtensor $\mathcal{T}_{[\mathbf{i}]}$ is
\begin{equation}\label{Entry}
(\mathcal{T}_{[\mathbf{i}]})_{\mathbf{j}} =  (\mathcal{T}_{[\mathbf{i}]})_{j_1\cdots j_m} =\mathcal{T}_{(\rho_{i_1}^{(1)}+j_1) \cdots (\rho_{i_m}^{(m)}+j_m)}.
\end{equation}

\begin{definition}
  Let $\mathcal{T}\in\mathbb{C}^{n_1\times\cdot\cdot\cdot\times n_m}$ be an $m$th-order tensor, $\mathbf{p}=\{p_1,\cdots,p_m\}$ be a permutation of $1:m$. The $\mathbf{p}$-transpositional tensor of $\mathcal{T}$, denoted by $\mathcal{T}^{<\mathbf{p}>}$, is  defined as
  $$
  (\mathcal{T}^{<\mathbf{p}>})_{\mathbf{p(\mathbf{j})}}=\mathcal{T}_{\mathbf{j}},
  $$
where for each $\mathbf{j}=\{j_1,...,j_m\}$, $\mathbf{p(\mathbf{j})}=\{j_{p_1},...,j_{p_m}\}$ is a $\mathbf{p}$-transposition of $\mathbf{j}$.
\end{definition}

\begin{lemma} \label{lem:jan24}
Let $\mathcal{T}\in\mathbb{C}^{n_1\times\cdot\cdot\cdot\times n_m}$ be a $b_1\times\cdot\cdot\cdot\times b_m$ block tensor defined by the partition (\ref{Seg}). Let $\mathbf{p}=\{p_1,\cdots,p_m\}$ be a permutation of $1:m$, $\mathbf{i}=\{i_1,...,i_m\}$. Then $$(\mathcal{T}^{<\mathbf{p}>})_{[\mathbf{p(i)}]}=(\mathcal{T}_{[\mathbf{i}]})^{<\mathbf{p}>}.$$
\end{lemma}
The proof of Lemma \ref{lem:jan24} is essentially the same as that of $\emph{Lemma 2.1}$ in \cite{RV2013}, hence it is omitted.

\section{Iterative methods for computing U-eigenpairs of non-symmetric  complex tensors}
In this section, we first introduce how to embed a non-symmetric complex tensor $\mathcal{A}$ into a symmetric complex tensor $\mathcal{S}$, and illustrate the relationship between the U-eigenpairs of $\mathcal{A}$ and the US-eigenpairs of $\mathcal{S}$. Then we propose three iterative algorithms to compute the eigenpairs of a non-symmetric tensor.

\subsection{The extended embedding operation}

In this subsection, we present the relationship between the U-eigenpairs of a non-symmetric complex tensor $\mathcal{A}$ and the US-eigenpairs of its symmetric embedding $\mathcal{S}$. We begin with the following definition.

\begin{definition}\label{def:embebdding}
Let $\mathcal{A}\in\mathbb{C}^{n_1\times\cdot\cdot\cdot\times n_m}$, $\mathbf{i}=\{i_1,...,i_m\}$. The $\mathbf{sym}(\cdot)$ operator means to construct an $n\times\cdot\cdot\cdot\times n$ block tensor denoted by $\mathcal{S}=\mathbf{sym}(\mathcal{A})\in \mathbb{C}^{n\times\cdots \times n}$, $n=n_1+\cdots+n_m$, and the $\mathbf{i}$-th block  $\mathcal{S}_{[\mathbf{i}]}\in\mathbb{C}^{n_{i_1}\times\cdot\cdot\cdot\times n_{i_m}}$ is defined as
$$\mathcal{S}_{[\mathbf{i}]}=\left\{
  \begin{array}{ll}
  \mathcal{A}^{<\mathbf{i}>},~~~~if~ \mathbf{i}~ is ~a ~permutation~ of ~$1:m$,\\
  ~~0,~~~~~~~else.
  \end{array}
\right.$$
$\mathcal{S}=\mathbf{sym}(\mathcal{A})$ is called the symmetric embedding of $\mathcal{A}$.
\end{definition}

\begin{theorem}
Given a tensor $\mathcal{A}\in\mathbb{C}^{n_1\times\cdot\cdot\cdot\times n_m}$, let  $\mathcal{S}=\mathbf{sym}(\mathcal{A})\in \mathbb{C}^{n\times\cdots\times n}$ be its symmetric embedding as given in Definition \ref{def:embebdding}. Then $\mathcal{S}$ is symmetric.
\end{theorem}

\begin{proof} The proof is similar to $\emph{Lemma 2.2}$ in \cite{RV2013}, hence it is omitted.
\end{proof}

We use a simple example to demonstrate symmetric embedding. Let $\mathcal{A}\in\mathbb{C}^{3\times4\times 5}$. Then  $\mathcal{S}=\mathbf{sym}(\mathcal{A})\in\mathbb{C}^{12\times12\times12}$. Moreover,
$$ \mathcal{S}_{[ijk]}=0\in\mathbb{C}^{n_i\times n_j\times n_k},\ \mathrm{for\ all \ } i,j,k\in\{1,2,3\},\ \mathrm{\ if\ } i=j \mathrm{\ or\ } i=k \mathrm{\ or\ } j=k, $$
and
$$\mathcal{S}_{[1 2 3]}=\mathcal{A}^{<1 2 3>}, \ \mathcal{S}_{[1 3 2]}=\mathcal{A}^{<1 3 2>}, \ \mathcal{S}_{[2 1 3]}=\mathcal{A}^{<2 1 3>}, $$$$ \mathcal{S}_{[2 3 1]}=\mathcal{A}^{<2 3 1>},\ \mathcal{S}_{[3 1 2]}=\mathcal{A}^{<3 1 2>}, \ \mathcal{S}_{[3 2 1]}=\mathcal{A}^{<3 2 1>}.$$

The following theorem illustrates the relationship between the U-eigenpairs of a complex non-symmetric tensor $\mathcal{A}$ and the US-eigenpairs of its symmetric embedding $\mathcal{S}=\mathbf{sym}(\mathcal{A})$.
\begin{theorem}\label{Th:eigenpair}
  Let $\mathcal{A}\in\mathbb{C}^{n_1\times\cdots\times n_m}$, $\mathcal{S}=\mathbf{sym}(\mathcal{A})$, $n=n_1+\cdots+ n_m$. Assume that $ \lambda_\mathcal{S}$ is a nonzero US-eigenvalue of $\mathcal{S}$ associated with the US-eigenvector $\mathbf{x}\in\mathbb{C}^{n}$. We partition $\mathbf{x}$ as $\mathbf{x}=(\mathbf{x}^{(1)\top},...,\mathbf{x}^{(m)\top})^\top$, $\mathbf{x}^{(i)}\in\mathbb{C}^{n_i}$ for all $i=1:m$. Then the following hold:\\
  (a) For $i=1,\cdots, m$, $\|\mathbf{x}^{(i)}\|=\frac{1}{\sqrt{m}}$, i.e., all $\mathbf{x}^{(i)}$ have the same norm $\frac{1}{\sqrt{m}}$.\\
  (b) Let $\lambda_\mathcal{A}=\frac{({\sqrt{m}})^m}{m!}\lambda_\mathcal{S}$. Then $\lambda_\mathcal{A}$ is a U-eigenvalue of $\mathcal{A}$ associated with the U-eigenvector $\{\sqrt{m}\mathbf{x}^{(1)},\ \cdots,\ \sqrt{m}\mathbf{x}^{(m)}\}$.
\end{theorem}

\begin{proof}
(a) Since $ \lambda_\mathcal{S}$ is a US-eigenvalue of $\mathcal{S}$ associated with the US-eigenvector $\mathbf{x}\in\mathbb{C}^{n}$, then we have
  \begin{equation}\label{Eq:SEigenPair}
   \langle \mathcal{S}, \mathbf{x}^{m-1}\rangle_{[i]}= \lambda_{\mathcal{S}}\mathbf{x}^{(i)*}, \ \ \ i=1,\cdots,m.
  \end{equation}
By the definition of the inner product of complex tensors, we have
\begin{eqnarray}
% \nonumber to remove numbering (before each equation)
\nonumber \langle \mathcal{S}, \mathbf{x}^{m-1}\rangle_{[i]} &=&\sum_{i_2,\cdots,i_m
  =1}^{m}  \langle\mathcal{S}_{[i i_2\cdots  i_m]}, \mathbf{x}^{(i_2)}\cdots \mathbf{x}^{(i_m)}\rangle\\
 \nonumber &=& \sum_{[i_2,\cdots,i_m]\in \mathbf{p}(1,2,...,i-1,i+1,...,m)}  \langle\mathcal{S}_{[i i_2\cdots  i_m]}, \mathbf{x}^{(i_2)}\cdots \mathbf{x}^{(i_m)}\rangle\\
 \nonumber &=&\sum_{[i_2,\cdots,i_m]\in \mathbf{p}(1,2,...,i-1,i+1,...,m)} \langle\mathcal{A}^{<i i_2\cdots i_m>}, \mathbf{x}^{(i_2)}\cdots \mathbf{x}^{(i_m)}\rangle\\
\label{eq1}   &=&(m-1)! \langle\mathcal{A}, \mathbf{x}^{(1)}\cdots \mathbf{x}^{(i-1)}\mathbf{x}^{(i+1)} \cdots \mathbf{x}^{(m)}\rangle
  \end{eqnarray}

Comparing the right-hand sides of (\ref{Eq:SEigenPair}) and (\ref{eq1}), we have that

\begin{equation}\label{eq2}
  (m-1)! \langle\mathcal{A}, \mathbf{x}^{(1)}\cdots \mathbf{x}^{(i-1)}\mathbf{x}^{(i+1)} \cdots \mathbf{x}^{(m)}\rangle=\lambda_{\mathcal{S}}\mathbf{x}^{(i)*}.
\end{equation}

It follows that
\begin{equation}\label{eq3}
    (m-1)!\langle \mathcal{A},\mathbf{x}^{(1)}\cdots\mathbf{x}^{(m)}\rangle=\lambda_\mathcal{S} \langle\mathbf{x}^{(i)*}, \mathbf{x}^{(i)}\rangle.
\end{equation}

On the other hand, there is
\begin{equation}\label{eq4}
\lambda_\mathcal{S}=\langle \mathcal{S}, \mathbf{x}^m\rangle=\sum^{m}_{i_1,\cdots,i_m=1}\langle \mathcal{S}_{[i_1,\cdots, i_m]},\mathbf{x}^{(i_1)}\cdots \mathbf{x}^{(i_m)}\rangle=m!\langle \mathcal{A},\mathbf{x}^{(1)}\cdots\mathbf{x}^{(m)}\rangle.
\end{equation}

Since $\lambda_\mathcal{S}\not=0$, by (\ref{eq3}) and (\ref{eq4}), we have that
\begin{equation}\label{eq5}
\langle\mathbf{x}^{(i)*}, \mathbf{x}^{(i)}\rangle=\frac{1}{m},~~~ i.e.~~~ \|\mathbf{x}^{(i)}\|=\frac{1}{\sqrt{m}}
\end{equation}
for all $i=1:m$.

(b) According to (\ref{eq2}) and (\ref{eq5}), we have $ \|\sqrt{m}\mathbf{x}^{(i)}\|=1$, and
\begin{equation}\label{1}
 \langle\mathcal{A}, (\sqrt{m}\mathbf{x}^{(1)})\cdots (\sqrt{m}\mathbf{x}^{(i-1)})(\sqrt{m}\mathbf{x}^{(i+1)}) \cdots (\sqrt{m}\mathbf{x}^{(m)})\rangle
  = \frac{({\sqrt{m}})^m\lambda_\mathcal{S}}{m!}\sqrt{m}\mathbf{x}^{(i)*}.
\end{equation}
By the definition of the U-eigenvalue of complex tensors, it follows that
\begin{equation}\label{eq7}
    \lambda_\mathcal{A}=\frac{({\sqrt{m}})^m}{m!}\lambda_\mathcal{S}
\end{equation}
is a U-eigenvalue of $\mathcal{A}$ associated with the eigenvectors $\{\sqrt{m}\mathbf{x}^{(1)},\ \cdots,\ \sqrt{m}\mathbf{x}^{(m)}\}$. This completes the proof. \ \ $\Box$
\end{proof}

\subsection{Iterative methods}
When an $n_1\times\cdots\times n_m$ tensor $\mathcal{A}$ is non-symmetric, we can use $\emph{Algorithm 4.1}$ in \cite{NB16} to compute the US-eigenpairs of its symmetric embedding $\mathcal{S}=\mathbf{sym}(\mathcal{A})$, and obtain the U-eigenpairs of the non-symmetric complex tensor $\mathcal{A}$ through $\emph{Theorem \ref{Th:eigenpair}}$. Hence, we have the following algorithm.

\begin{algo}\label{algo:1}
Computing the U-eigenpairs of an $n_1\times\cdots\times n_m$ non-symmetric complex tensor $\mathcal{A}$.

{\bf Step 1 (Initial step)}:
Let $\mathcal{S}=\mathbf{sym}(\mathcal{A})$, and $n=n_1+\cdots+n_m$. Choose a starting point $\mathbf{x}_0\in\mathbb{C}^n$ with $\|\mathbf{x}_0\|=1$, and $0<\alpha_\mathcal{S}\in \mathbb{R}$. Let $\lambda_0=\mathcal{S}^*\mathbf{x}_0^m$.

{\bf Step 2 (Iterating step)}:

\hskip 0.5 in {\bf for} $k=1,2,\cdots$, {\bf do}
\begin{eqnarray*}\label{eqalgorithm1.1}
\label{eqalgorithm1.1} \hat{\mathbf{x}}_k&=&\lambda_{{k-1}}\mathcal{S}\mathbf{x}^{*m-1}_{k-1}+\alpha_\mathcal{S}\mathbf{x}_{k-1},\\
\label{eqalgorithm1.2} \mathbf{x}_k &=&\hat{\mathbf{x}}_k/\|\hat{\mathbf{x}}_k\|,\\
\label{eqalgorithm1.3} \lambda_k &=&\mathcal{S}^*\mathbf{x}_k^m.
\end{eqnarray*}

\hskip 0.5 in {\bf end for.}

{\bf return}:

\hskip 0.5 in US-eigenvalue $\lambda_\mathcal{S}=|\lambda_k| $, US-eigenvector $\mathbf{x}=(\frac{\lambda_\mathcal{S}}{\lambda_k})^{1/m}\mathbf{x}_k$.

\hskip 0.5 in Let $\mathbf{x}=(\mathbf{x}^{(1)\top},...,\mathbf{x}^{(m)\top})^\top$, $\mathbf{x}^{(i)}\in \mathbb{C}^{n_i},$ for all $i=1:m$.

\hskip 0.5 in U-eigenvalue $\lambda_\mathcal{A}=\frac{({\sqrt{m}})^m}{m!}\lambda_\mathcal{S}$.

\hskip 0.5 in U-eigenvector $\{\sqrt{m}\mathbf{x}^{(1)},\ \cdots,\ \sqrt{m}\mathbf{x}^{(m)}\}$.
\end{algo}

Given a non-symmetric tensor $\mathcal{A}$, the size of the symmetric embedding tensor $\mathcal{S}=\mathbf{sym}(\mathcal{A})$ is much larger than that of the tensor $\mathcal{A}$ itself. This affects the computational efficiency of Algorithm \ref{algo:1} proposed above. Motivated by this,  we propose a new iterative algorithm.

\begin{algo}\label{algo:2}
Computing the U-eigenpairs of an $n_1\times\cdots\times n_m$ non-symmetric complex tensor $\mathcal{A}$.

{\bf Step 1 (Initial step)}:
Choose starting points $\hat{\mathbf{x}}_0^{(i)}\in\mathbb{C}^{n_i}$ with $||\hat{\mathbf{x}}_0^{(i)}||\not=0$ for all $i=1:m$.
Let $\mathbf{x}_0^{(i)}=\hat{\mathbf{x}}_0^{(i)}/\sqrt{\sum_{j=1}^m ||\hat{\mathbf{x}}_0^{(j)}||^2}$ for all $i=1:m$,
$\lambda_{0}=\langle\mathcal{A}, \mathbf{x}^{(1)}_0\cdots\mathbf{x}^{(m)}_0\rangle$. Choose $0<\alpha_{\mathcal{A}}\in \mathbb{R}$.

{\bf Step 2 (Iterating step)}:

\hskip 0.4 in {\bf for} $k=1,2,\cdots$, {\bf do}

\hskip 0.5 in {\bf for} $i=1,2,\cdots, m$, {\bf do}

\centerline{ $ \hskip 0.5 in \hat{\mathbf{x}}^{(i)}_k = \lambda_{k-1}\mathcal{A}\mathbf{x}^{(1)*}_{k-1}\cdots\mathbf{x}^{(i-1)*}_{k-1}\mathbf{x}^{(i+1)*}_{k-1}\cdots\mathbf{x}^{(m)*}_{k-1}+  \alpha_{\mathcal{A}}\mathbf{x}^{(i)}_{k-1} $.}

\hskip 0.5 in {\bf end for.}

\hskip 0.5 in {\bf for} $i=1,2,\cdots, m$, {\bf do}

$$ \mathbf{x}_k^{(i)} = \frac{\hat{\mathbf{x}}_k^{(i)}}{\sqrt{\sum_{j=1}^m ||\hat{\mathbf{x}}_k^{(j)}||^2}}. %\\\nonumber \\
$$

\hskip 0.5 in {\bf end for.}

\centerline{ $\lambda_k = \langle\mathcal{A}, \mathbf{x}^{(1)}_k\cdots\mathbf{x}^{(m)}_k\rangle  $.}

\hskip 0.4 in{\bf end for.}

{\bf return}:

\hskip 0.5 in U-eigenvalue $\lambda_\mathcal{A}=({\sqrt{m}})^m|\lambda_k|$.

\hskip 0.5 in U-eigenvector $\{\sqrt{m} (\frac{|\lambda_k|}{\lambda_k})^{1/m} \mathbf{x}^{(1)},\ \cdots,\ \sqrt{m} (\frac{|\lambda_k|}{\lambda_k})^{1/m} \mathbf{x}^{(m)}\}$.
\end{algo}

The following theorem establishes the relationship between Algorithms \ref{algo:1} and \ref{algo:2}.

\begin{theorem}\label{Th:non-symmetric eigenpair}
Let $\mathcal{A}\in\mathbb{C}^{n_1\times\cdots\times n_m}$, $\mathcal{S}=\mathbf{sym}(\mathcal{A})$, $\alpha_\mathcal{S} = m!(m-1)!\ \alpha_\mathcal{A}$, $n=n_1+\cdots+ n_m$. Choose points  $\hat{\mathbf{x}}^{(i)}_0\in\mathbb{C}^{n_{i}}$ with $\|\hat{\mathbf{x}}^{(i)}_0\|\not=0$ for all $i=1:m$. Let $\mathbf{x}_0^{(i)}=\hat{\mathbf{x}}_0^{(i)}/\sqrt{\sum_{j=1}^m ||\hat{\mathbf{x}}_0^{(j)}||^2}$ for all $i=1:m$, and denote
 $\mathbf{x}_0=(\mathbf{x}^{(1)\top}_0,...,\mathbf{x}^{(m)\top}_0)^\top$.
Let  $\mathbf{x}_0$ be the starting point of Algorithm \ref{algo:1}, and assume that $\lambda_{\mathcal{S}_k}$ and $\mathbf{x}_k$ are obtained by the $k$-th iteration of Algorithm \ref{algo:1}. Similarly, let $\{\mathbf{x}^{(1)}_0,\cdots,\mathbf{x}^{(m)}_0\}$ be the starting point of Algorithm \ref{algo:2}, and assume that $\lambda_{\mathcal{A}_k}$ and $\{\mathbf{x}^{(1)}_k,\cdots,\mathbf{x}^{(m)}_k\}$ are obtained by the $k$-th iteration of Algorithm \ref{algo:2}. Then
 $\mathbf{x}_k=(\mathbf{x}^{(1)\top}_k,...,\mathbf{x}^{(m)\top}_k)^\top$ and  $\lambda_{\mathcal{S}_k}= m! \lambda_{\mathcal{A}_k}$.
\end{theorem}
\begin{proof}
Partition $\mathbf{x}_k$ as $\mathbf{x}_k=(\bar{\mathbf{x}}^{(1)\top}_k,...,\bar{\mathbf{x}}^{(m)\top}_k)^\top$ with $\bar{\mathbf{x}}^{(i)}_k\in \mathbb{C}^{n_i}$. Similar to (\ref{eq1}), we have
\begin{equation}
% \nonumber to remove numbering (before each equation)
\label{eq122}  \langle \mathcal{S}, \mathbf{x}_k^{m-1}\rangle_{[i]} = (m-1)! \langle\mathcal{A}, \bar{\mathbf{x}}_k^{(1)}\cdots \bar{\mathbf{x}}_k^{(i-1)}\bar{\mathbf{x}}_k^{(i+1)} \cdots \bar{\mathbf{x}}_k^{(m)}\rangle.
  \end{equation}
\begin{equation}
% \nonumber to remove numbering (before each equation)
\label{eq123}  \lambda_{\mathcal{S}_k}=\langle \mathcal{S}, \mathbf{x}_k^{m}\rangle = m! \langle\mathcal{A}, \bar{\mathbf{x}}_k^{(1)}\cdots  \bar{\mathbf{x}}_k^{(m)}\rangle.
  \end{equation}
In the following, we will use the mathematical induction to prove Theorem \ref{Th:non-symmetric eigenpair}.
Let k=0. By the assumption of $\mathbf{x}_0$ and $\{\mathbf{x}_0^{(1)}, \mathbf{x}_0^{(2)}, \cdots,\  \mathbf{x}_0^{(m)}\}$, we have that
$$\bar{\mathbf{x}}^{(i)}_0=\mathbf{x}^{(i)}_0\ (i=1,2, \cdots, m), \  \mathbf{x}_0=(\mathbf{x}^{(1)\top}_0,...,\mathbf{x}^{(m)\top}_0)^\top.$$
By (\ref{eq123}),
$$ \lambda_{\mathcal{S}_0}= m! \langle\mathcal{A}, \bar{\mathbf{x}}_0^{(1)}\cdots  \bar{\mathbf{x}}_0^{(m)}\rangle = m!\langle\mathcal{A}, {\mathbf{x}}_0^{(1)}\cdots  {\mathbf{x}}_0^{(m)}\rangle=m! \lambda_{\mathcal{A}_0} . $$
Hence, the result holds for $k=0$. Assume that the result holds for  $k-1$. That is
$$\bar{\mathbf{x}}^{(i)}_{k-1}=\mathbf{x}^{(i)}_{k-1}\ (i=1,2, \cdots, m), \  \mathbf{x}_{k-1}=(\mathbf{x}^{(1)\top}_{k-1},...,\mathbf{x}^{(m)\top}_{k-1})^\top,  \lambda_{\mathcal{S}_{k-1}}= m! \lambda_{\mathcal{A}_{k-1}}.$$
Partition $\hat{\mathbf{x}}_k$ as $\hat{\mathbf{x}}_k = \{\breve{\mathbf{x}}_k^{(1)},\ \cdots,\ \breve{\mathbf{x}}_k^{(m)}  \} $ with $\breve{\mathbf{x}}_k^{(i)}\in \mathbb{C}^{n_i} $. Then by Algorithm \ref{algo:1}, we have
\begin{eqnarray*}
% \nonumber to remove numbering (before each equation)
  \hat{\mathbf{x}}_k &=& \lambda_{\mathcal{S}_{k-1}}\mathcal{S}\mathbf{x}^{*m-1}_{k-1}+\alpha_\mathcal{S}\mathbf{x}_{k-1}
   = \lambda_{\mathcal{S}_{k-1}} \langle\mathcal{S}, \mathbf{x}^{m-1}_{k-1}\rangle^*+\alpha_\mathcal{S}\mathbf{x}_{k-1}.
\end{eqnarray*}
Then,
\begin{eqnarray*}
% \nonumber to remove numbering (before each equation)
\breve{\mathbf{x}}^{(i)}_k &=& \lambda_{\mathcal{S}_{k-1}} \langle\mathcal{S}, \mathbf{x}^{m-1}_{k-1}\rangle^*_{[i]}+\alpha_\mathcal{S}\bar{\mathbf{x}}^{(i)}_{k-1}\\
 &=& \lambda_{\mathcal{S}_{k-1}} (m-1)!\ \langle\mathcal{A}, \mathbf{x}_{k-1}^{(1)}\cdots \mathbf{x}_{k-1}^{(i-1)}\mathbf{x}_{k-1}^{(i+1)} \cdots \mathbf{x}_{k-1}^{(m)}\rangle^*+\alpha_\mathcal{S}\mathbf{x}^{(i)}_{k-1}.\\
% &=& m!(m-1)!\ \lambda_{\mathcal{A}_{k-1}}\ \langle\mathcal{A}, \mathbf{x}_{k-1}^{(1)}\cdots \mathbf{x}_{k-1}^{(i-1)}\mathbf{x}_{k-1}^{(i+1)} \cdots \mathbf{x}_{k-1}^{(m)}\rangle^*+m!(m-1)!\  \alpha_\mathcal{A}\mathbf{x}^{(i)}_{k-1}.\\
 &=& m!(m-1)!\ \left( \lambda_{\mathcal{A}_{k-1}}\ \langle\mathcal{A}, \mathbf{x}_{k-1}^{(1)}\cdots \mathbf{x}_{k-1}^{(i-1)}\mathbf{x}_{k-1}^{(i+1)} \cdots \mathbf{x}_{k-1}^{(m)}\rangle^*+ \alpha_\mathcal{A}\mathbf{x}^{(i)}_{k-1}\right)\\
 &=& m!(m-1)!\ \hat{\mathbf{x}}^{(i)}_{k}.
\end{eqnarray*}
Hence,
$$\hat{\mathbf{x}}_k = m!(m-1)! (\hat{\mathbf{x}}^{(1)\top}_{k}, \hat{\mathbf{x}}^{(2)\top}_{k}, \cdots, \hat{\mathbf{x}}^{(m)\top}_{k})^\top.$$
It follows that
$$\mathbf{x}_k =\frac{\hat{\mathbf{x}}_k}{||\hat{\mathbf{x}}_k||} =\frac{(\hat{\mathbf{x}}^{(1)\top}_{k}, \hat{\mathbf{x}}^{(2)\top}_{k}, \cdots, \hat{\mathbf{x}}^{(m)\top}_{k})^\top}{\sqrt{\sum_{j=1}^m || \hat{\mathbf{x}}^{(j)}_{k} ||^2 }} = (\mathbf{x}^{(1)\top}_{k}, \mathbf{x}^{(2)\top}_{k}, \cdots, \mathbf{x}^{(m)\top}_{k})^\top.$$
By (\ref{eq123}),
$$ \lambda_{\mathcal{S}_k}= m! \langle\mathcal{A}, \bar{\mathbf{x}}_k^{(1)}\cdots  \bar{\mathbf{x}}_k^{(m)}\rangle = m!\langle\mathcal{A}, {\mathbf{x}}_k^{(1)}\cdots  {\mathbf{x}}_k^{(m)}\rangle=m! \lambda_{\mathcal{A}_k} . $$
Hence, the result holds for  $k$. This completes the proof. \ \ \ $\Box$
\end{proof}

{\it Remark 1.} The convergence of Algorithm \ref{algo:1} has already been proved in \cite{NB16}. By Theorem \ref{Th:non-symmetric eigenpair}, Algorithm \ref{algo:2} is also convergent.

{\it Remark 2.}  In \cite{QZN2017} an algorithm is proposed to compute the U-eigenpairs of non-symmetric complex tensors.  The difference between it and Algorithm \ref{algo:2} is  that they use different normalization conditions. To be specific, in Algorithm \ref{algo:2}, the normalization condition  is $\mathbf{x}_k^{(i)} = \frac{\hat{\mathbf{x}}_k^{(i)}}{\sqrt{\sum_{j=1}^m ||\hat{\mathbf{x}}_k^{(j)}||^2}} $, while the normalization condition  used in the algorithm of  \cite{QZN2017} is $\mathbf{x}_k^{(i)} = \frac{\hat{\mathbf{x}}_k^{(i)}}{ ||\hat{\mathbf{x}}_k^{(i)}||}$. This is a key difference as the new normalization condition enables us to establish the convergence of Algorithm \ref{algo:2}.

Inspired by the well-known Gauss-Seidel method \cite{Zeng2004}, we propose the following algorithm which may improve computational efficiency of   Algorithm \ref{algo:2}.

\begin{algo}\label{algo:3}
The Gauss-Seidel method for computing U-eigenpairs of an $n_1\times\cdots\times n_m$ non-symmetric complex tensor $\mathcal{A}$.

{\bf Step 1 (Initial step)}:
Choose starting points $\mathbf{x}_0^{(i)}\in\mathbb{C}^{n_i}$ with $||\mathbf{x}_0^{(i)}||=1$ for all $i=1:m$. Let $\lambda_{0}=\langle\mathcal{A}, \mathbf{x}^{(1)}_0\cdots\mathbf{x}^{(m)}_0\rangle$. Choose $0<\alpha_{\mathcal{A}}\in \mathbb{R}$.

{\bf Step 2 (Iterating step)}:

\hskip 0.5 in {\bf for} $k=1,2,\cdots$, {\bf do}

\hskip 0.7 in {\bf for} $i=1,2,\cdots, m$, {\bf do}

\hskip 0.8 in \ \ \ $ \hat{\mathbf{x}}^{(i)}_k = \lambda_{k-1}\mathcal{A}\mathbf{x}^{(1)*}_{k} \cdots \mathbf{x}^{(i-1)*}_{k} \mathbf{x}^{(i+1)*}_{k-1}\cdots\mathbf{x}^{(m)*}_{k-1} + \alpha_{\mathcal{A}}\mathbf{x}^{(i)}_{k-1}$,

\hskip 0.8 in \ \ \  $ \mathbf{x}^{(i)}_k = \hat{\mathbf{x}}^{(i)}_k/ \|\hat{\mathbf{x}}^{(i)}_k\| $.

\hskip 0.7 in {\bf end for.}

\centerline {$ \lambda_k =\langle\mathcal{A}, \mathbf{x}^{(1)}_k\cdots\mathbf{x}^{(m)}_k\rangle $.}

\hskip 0.5 in {\bf end for.}

{\bf return}:

\hskip 0.5 in U-eigenvalue $\lambda_\mathcal{A}=|\lambda_k|$, U-eigenvector $\mathbf{x}^{(i)}=(\frac{\lambda_\mathcal{A}}{\lambda_k})^{1/m}\mathbf{x}^{(i)}_k$.
\end{algo}

\section{Numerical examples}
The first six examples compute the GME of pure states by Algorithms \ref{algo:1}-\ref{algo:3}, respectively. We compare their CPU time and iteration step. In the last two examples, we compute U-eigenvalues of non-symmetric complex tensors by Algorithm \ref{algo:3} and the network method \cite{CQWZ18}, respectively, and compare their CPU time and iteration step.

The computations are implemented in Mathematica 8.0 on a Microsoft Win10 Laptop with 8GB memory and Intel(R) i5 CPU 2.40GHZ. In the first six examples, we take 10 starting points to get the results, and obtain the maximum eigenvalue from them. Four decimal digits are presented, and two decimal digits are presented for CPU time.

\begin{example}\label{exam1}
  (\cite[Example 6]{hqz12}) Consider a non-symmetric 3-partite state $$|\psi\rangle= \sqrt{\frac{1}{3}}|001\rangle+\sqrt{\frac{2}{3}}|100\rangle.$$
 It corresponds to a $2\times2\times2$ non-symmetric tensor $\mathcal{A}$ with nonzero entries $\mathcal{A}_{112}=\sqrt{\frac{1}{3}},\ \mathcal{A}_{211}=\sqrt{\frac{2}{3}}$.

 We apply Algorithms \ref{algo:1}-\ref{algo:3} to get the largest U-eigenvalue $\lambda_{\mathcal{A}}$ of $\mathcal{A}$ and the GME of the state $|\psi\rangle$ respectively. The iteration is terminated when the numerical error is less than $10^{-9}$. The computational results are shown in Table \ref{Tb1}.

\begin{table}[!h]
\centering
\begin{tabular}{|l|c|c|c|c|}

\hline
Algorithms &$\lambda_{\mathcal{A}}$ & GME  &Time(sec)\\
\hline
Algorithm \ref{algo:1}& 0.8165& 0.6058&2.75\\
\hline
Algorithm \ref{algo:2}&0.8165& 0.6058&1.02\\
\hline
Algorithm \ref{algo:3}&0.8165& 0.6058&0.23\\
\hline
\end{tabular}
\caption{Computational results for Example \ref{exam1}}\label{Tb1}
\end{table}

The GME of the state $|\psi\rangle$ is $0.6058$ with the closest product state calculated by Algorithm \ref{algo:3}
$$|\phi\rangle = |\phi_1\rangle \times |\phi_2\rangle \times |\phi_3\rangle,$$
where
$$|\phi_1\rangle = (\ \ 0.8350-0.5502\mathrm{i})|1\rangle,$$
$$|\phi_2\rangle = (-0.3980+0.9174\mathrm{i})|0\rangle, $$
$$|\phi_3\rangle = (\ \ 0.1724-0.9850\mathrm{i})|0\rangle.$$

We can observe from Table \ref{Tb1} that Algorithms \ref{algo:1}-\ref{algo:3} obtain the same value of GME. However, Algorithm \ref{algo:1} takes  more time than Algorithms \ref{algo:2} and \ref{algo:3}.
 \end{example}

%%%%%%%%%%%%%%%%%%%%%%%%%%%%%%%%%%%%%%%%

\begin{example}\label{exam2}
  (\cite[Example 9]{QZN2017}) Consider a non-symmetric 3-partite state $$|\psi\rangle= \sqrt{\frac{1}{6}}(|000\rangle+|101\rangle+|012\rangle+|110\rangle+|021\rangle+|122\rangle).$$
 It corresponds to a $2\times3\times3$ non-symmetric tensor $\mathcal{A}$ with nonzero entries $\mathcal{A}_{111}=\mathcal{A}_{212}=\mathcal{A}_{123}=\mathcal{A}_{221}=\mathcal{A}_{132} =\mathcal{A}_{233}=\sqrt{\frac{1}{6}}$.

 We apply Algorithms \ref{algo:1}-\ref{algo:3} to get the maximum U-eigenvalue $\lambda_{\mathcal{A}}$ of $\mathcal{A}$ and the GME of the state $|\psi\rangle$ respectively. The iteration is terminated when the numerical error is less than $10^{-9}$. The computational results are shown in Table \ref{Tb2}.

\begin{table}[!htbp]
\centering
\begin{tabular}{|l|c|c|c|}

\hline
Algorithms &$\lambda_{\mathcal{A}}$ & GME  &Time(sec)\\
\hline
Algorithm \ref{algo:1}& 0.5774& 0.9194&5.97\\
\hline
Algorithm \ref{algo:2}& 0.5774& 0.9194&1.12\\
\hline
Algorithm \ref{algo:3}& 0.5774& 0.9194&0.19\\
\hline
\end{tabular}
\caption{Computational results for Example \ref{exam2}}\label{Tb2}
\end{table}

The GME of the state $|\psi\rangle$ is $0.9194$ with the closest  product state calculated by Algorithm \ref{algo:3}
$$|\phi\rangle = |\phi_1\rangle \times |\phi_2\rangle \times |\phi_3\rangle,$$
where
\begin{eqnarray}
% \nonumber to remove numbering (before each equation)
\nonumber |\phi_1\rangle &=&(-0.0336+0.7063\mathrm{i})|0\rangle+(\ \ 0.6285 -0.3240\mathrm{i})|1\rangle,\\
\nonumber |\phi_2\rangle &=&(-0.0353-0.5763\mathrm{i})|0\rangle+( 0.5167+0.2576\mathrm{i})|1\rangle-(0.4814-0.3187\mathrm{i})|2\rangle,\\
\nonumber |\phi_3\rangle &=&(\ \ 0.5773 +0.0079\mathrm{i})|0\rangle-(0.2955-0.4960\mathrm{i})|1\rangle-(0.2818+0.5039\mathrm{i})|2\rangle.
\end{eqnarray}

Compared to Example \ref{exam1}, we increase the dimension of the quantum state $|\psi\rangle$, and the corresponding $\mathbf{sym}(\mathcal{A})$ obtained by Algorithm \ref{algo:1} is an $8\times8\times8$ symmetric tensor. We can observe from Table \ref{Tb2} that Algorithms \ref{algo:1}-\ref{algo:3} get the same GME value; however, the time Algorithm \ref{algo:1} takes is much longer than that of Algorithms \ref{algo:2} and \ref{algo:3}.
\end{example}

%%%%%%%%%%%%%%%%%%%%%%%%%%%%%%%%%%%%%%%%
\begin{example}\label{exam3}
 (\cite[(37)]{MIZ2016}) Consider the 5-qubit AME state
 \begin{eqnarray*}
 % \nonumber to remove numbering (before each equation)
   |\psi\rangle &=& \frac{1}{ 2\sqrt{2}}(|00000\rangle+|00011\rangle+|01100\rangle -|01111\rangle +|11010\rangle \\
    & &  +|11001\rangle +|10110\rangle-|10101\rangle).
 \end{eqnarray*}
It corresponds to a $2\times2\times2\times2\times2$ non-symmetric tensor $\mathcal{A}$ with nonzero entries $\mathcal{A}_{11111}=\mathcal{A}_{11122}=\mathcal{A}_{12211}=\mathcal{A}_{22121}=\mathcal{A}_{22112} =\mathcal{A}_{21221}=\frac{1}{ 2\sqrt{2}},\ \mathcal{A}_{12222}=\mathcal{A}_{21212}=-\frac{1}{ 2\sqrt{2}}$.

 We apply Algorithms \ref{algo:1}-\ref{algo:3} to get the maximum U-eigenvalue $\lambda_{\mathcal{A}}$ of $\mathcal{A}$ and the GME of the state $|\psi\rangle$ respectively. The iteration is terminated when the numerical error is less than $10^{-9}$. The computational results are shown in Table \ref{Tb3}.

\begin{table}[!htbp]
\centering
\begin{tabular}{|l|c|c|c|}

\hline
Algorithms &$\lambda_{\mathcal{A}}$ & GME   &Time(sec)\\
\hline
Algorithm \ref{algo:1} &0.3626& 1.1291&2628.76\\
\hline
Algorithm \ref{algo:2} &0.3626& 1.1291&14.89\\
\hline
Algorithm \ref{algo:3} &0.3626& 1.1291&2.42\\
\hline
\end{tabular}
\caption{Computational results for Example \ref{exam3}}\label{Tb3}
\end{table}
The GME of the state $|\psi\rangle$ is $1.1291$ with the closest product state calculated by Algorithm \ref{algo:3}
$$|\phi\rangle = |\phi_1\rangle \times |\phi_2\rangle \times |\phi_3\rangle \times |\phi_4\rangle \times |\phi_5\rangle,$$
where
\begin{eqnarray}
% \nonumber to remove numbering (before each equation)
\nonumber |\phi_1\rangle &=&(-0.1455 +0.4361\mathrm{i})|0\rangle+(-0.7944+0.3969\mathrm{i})|1\rangle,\\
\nonumber |\phi_2\rangle &=&(-0.0940 -0.4500\mathrm{i})|0\rangle+(-0.7431-0.4863\mathrm{i})|1\rangle,\\
\nonumber |\phi_3\rangle &=&(\ \ 0.4913+0.7398\mathrm{i})|0\rangle+(-0.4506-0.0910\mathrm{i})|1\rangle,\\
\nonumber |\phi_4\rangle &=&(\ \ 0.8856+0.0669\mathrm{i})|0\rangle+(\ \ 0.2996+0.3486\mathrm{i})|1\rangle,\\
\nonumber |\phi_5\rangle &=&(-0.1751-0.4251\mathrm{i})|0\rangle+(\ \ 0.3415-0.8198\mathrm{i})|1\rangle.
\end{eqnarray}

Compared to Examples \ref{exam1} and \ref{exam2}, we increased the order of the quantum state $|\psi\rangle$, and the corresponding tensor $\mathbf{sym}(\mathcal{A})$ is a $10\times10\times10\times10\times10$ symmetric tensor. We can observe from Table \ref{Tb3} that although the three algorithms give us the same GME value. However, it takes Algorithm \ref{algo:1} much longer time to get the results.
\end{example}

It can be seen from the above examples that, when the order and dimension of the quantum state increase, the size of the symmetric tensor used in Algorithm \ref{algo:1} is getting more and more large, which directly leads to a decrease in the efficiency of Algorithm \ref{algo:1}. Hence,  in the following examples we only compare Algorithms \ref{algo:2} and \ref{algo:3}.

%%%%%%%%%%%%%%%%%%%%%%%%%%%%%%%%%%%%%%%%
\begin{example}
(\cite[Example 4.3]{NZZ2017})\label{exam4}
    Consider a non-symmetric 3-partite state
   $$|\psi\rangle= \sum_{i_1, i_2, i_3=1}^n \frac{\cos(i_1-i_2+i_3)+\sqrt{-1}\  \sin(i_1+i_2-i_3)}{\sqrt{n^3}}|(i_1-1)(i_2-1)(i_3-1)\rangle.$$
It corresponds to an $n \times n\times n$ non-symmetric tensor $\mathcal{A}$ with $$\mathcal{A}_{i_1i_2i_3}= \frac{\cos(i_1-i_2+i_3)+\sqrt{-1}\  \sin(i_1+i_2-i_3)}{\sqrt{n^3}}.$$

\begin{table}[!h]
\centering
\begin{tabular}{|c|c|c|c|c|}

\hline
n&$\lambda_{\mathcal{A}}$ & GME  & Algo \ref{algo:2}(sec) & Algo \ref{algo:3}(sec) \\
\hline
2 & 0.8895& 0.4701&0.47&0.09\\
\hline
5& 0.7815& 0.6611&5.90&1.62\\
\hline
10& 0.7072& 0.7652&51.59&8.74\\
\hline
15& 0.7243& 0.7425&173.06&11.53\\
\hline
20& 0.7175& 0.7516&969.73&96.07\\
\hline
50& 0.7087& 0.7632&127332.83&2070.27\\
\hline
\end{tabular}
\caption{Computational results for Example \ref{exam4}}\label{Tb4}
\end{table}

For a range of values of $n$ from 2 to 50, we apply Algorithms \ref{algo:2} and \ref{algo:3} to get the maximum U-eigenvalue $\lambda_{\mathcal{A}}$ of $\mathcal{A}$ and the GME of  the state $|\psi\rangle$. The iteration is terminated when the numerical error is less than $10^{-9}$. The computational results are presented in Table \ref{Tb4}.

We can observe from Table \ref{Tb4} that both Algorithm \ref{algo:2} and Algorithm \ref{algo:3} can obtain the same GME value of the state $|\psi\rangle$ for different $n$. However, as $n$ increases,  Algorithm \ref{algo:3} becomes much more efficient than Algorithm \ref{algo:2}.
\end{example}

%%%%%%%%%%%%%%%%%%%%%%%%%%%%%%%%%%%%%%%%
\begin{example}\label{exam5}
  (Random Examples)
   In this example, we randomly generate quantum pure states, and compute their GME by means of Algorithms \ref{algo:2} and \ref{algo:3}.
   For a range of values of order $m$ from 4 to 5, we apply both Algorithm \ref{algo:2} and Algorithm \ref{algo:3} to get the maximum U-eigenvalue $\lambda_{\mathcal{A}}$ of $\mathcal{A}$ and the GME of the state $|\psi\rangle$. The iteration is terminated when the numerical error is less than $10^{-9}$.  The computational results are presented in Tables \ref{Tb5} and \ref{Tb6} respectively.

\begin{table}[!h]
\centering
\begin{tabular}{|c|c|c|c|c|}

\hline
$(n_1\times n_2\times n_3\times n_4)$&$\lambda_{\mathcal{A}}$ & GME & Algo \ref{algo:2}(sec) & Algo \ref{algo:3}(sec) \\
\hline
$3\times3 \times 3\times 3$&0.4854&  1.0145&21.40& 2.84\\
\hline
$5\times5 \times 5\times 5$& 0.2555& 1.2203& 242.13&66.06\\
\hline
$2\times5 \times 8\times 15$&0.2240 & 1.2458&656.30&131.14\\
\hline
\end{tabular}
\caption{Computational results for Example \ref{exam5} with $m=4$}\label{Tb5}
\end{table}

\begin{table}[!h]
\centering
\begin{tabular}{|c|c|c|c|c|}

\hline
$(n_1\times n_2\times n_3\times n_4\times n_5)$&$\lambda_{\mathcal{A}}$ & GME & Algo \ref{algo:2}(sec)& Algo \ref{algo:3}(sec)\\
\hline
$2\times2 \times 2\times 2\times 2$& 0.5475&0.9513&9.70&1.59\\
\hline
$8\times2 \times 3\times 5\times 4$&0.2164 & 1.2519& 720.05 &103.35 \\
\hline
$10\times3 \times 15\times 2\times 5$&0.1322 &1.3175&12194.04&1406.26\\
\hline
\end{tabular}
\caption{Computational results for Example \ref{exam5} with $m=5$} \label{Tb6}
\end{table}

We can observe from Tables \ref{Tb5} and \ref{Tb6} that the same GME values of the state $|\psi\rangle$ are obtained by both    Algorithm \ref{algo:2} and Algorithm \ref{algo:3} within a reasonable time. When the order $m$ or dimension $n$ increases, it will take more time to find the GME of  the state $|\psi\rangle$. Moreover, it is clear that Algorithm \ref{algo:3} is always more efficient than Algorithm \ref{algo:2}.
\end{example}
%%%%%%%%%%%%%%%%%%%%%%%%%%%%%%%%%%%%%%%%
\begin{example}(\cite[(31)]{GBZ16}, \cite[Example 11]{QZN2017})\label{exam6}
Given a  $3\times3\times3\times3\times3\times2$  pure state
\begin{eqnarray*}
\ket{\psi} &=&\frac{1}{3\sqrt{2}}(\ket{000000}+\ket{001121}+\ket{010220}\\
&& + \ket{012011}+\ket{021210}+\ket{022101}\\
&&+\ket{111110}+\ket{112201}+\ket{121000}\\
&&+\ket{120121}+\ket{102020}+\ket{100211}\\
&&+\ket{222220}+\ket{220011}+\ket{202110}\\
&&+\ket{201201}+\ket{210100}+\ket{211021}),
\end{eqnarray*}
Algorithms \ref{algo:2} and \ref{algo:3} give the same GME value 1.2364, which is the same as that of \cite[Example 11]{QZN2017}. The time  Algorithms \ref{algo:2} takes is 1298.95 seconds, while  the time Algorithms \ref{algo:3} takes is merely 9.70 seconds.
\end{example}

We end this section with two final examples to compare Algorithm \ref{algo:3} with the neural network method peoposed by Che et al.\cite{CCW17,CQWZ18} for computing the U-eigenvalues of complex tensors.

%%%%%%%%%%%%%%%%%%%%%%%%%%%%%%%%%%%%%%%%
\begin{example}\label{exam7}
    Consider a $10\times8\times5\times7$ non-symmetric complex tensor $\mathcal{A}$ with nonzero entries
$\mathcal{A}_{8726}=\frac{1}{\sqrt{6}}, \mathcal{A}_{9543}=\frac{1}{\sqrt{3}}, \mathcal{A}_{1221}=\frac{1}{\sqrt{6}}\mathrm{i}, \mathcal{A}_{3812}=-\frac{1}{\sqrt{3}}$.
	
We apply Algorithms \ref{algo:3} and the neural network method \cite{CQWZ18} to compute the  U-eigenvalue $\lambda_{\mathcal{A}}$ of $\mathcal{A}$, respectively, and compare their iteration step and CPU time. The iteration is terminated when the numerical error is less than $10^{-9}$. The computational results are shown in Table \ref{Tb7}. It can be observed that these two algorithms obtain the same U-eigenvalue, and Algorithm \ref{algo:3} takes less CPU time and iteration step than the network method.
	
	\begin{table}[!htbp]
		\centering
		\begin{tabular}{|l|c|c|c|}
			
			\hline
			Algorithm &$\lambda_{\mathcal{A}}$ & Iteration Step   & CPU Time(sec)\\
			\hline
			Neural network method &0.5774 & 133 & 25.20\\
			\hline
			Algorithm \ref{algo:3} & 0.5774 & 25 & 5.60\\
			\hline
		\end{tabular}
		\caption{Computational results for Example \ref{exam7}}\label{Tb7}
	\end{table}
	
\end{example}

%%%%%%%%%%%%%%%%%%%%%%%%%%%%%%%%%%%%%%%%
\begin{example}\label{exam8}
Here, we will show the the performance of these two algorithms in the computation of U-eigenvalues for different random complex tensors. The computation results are presented in Table \ref{Tb8}, where T1, T2, Iter1 and Iter2  represent for the CPU time for  the neural network method, the CPU time for  Algorithm \ref{algo:3}, the iteration steps for the  the neural network method, and the iteration steps for Algorithm \ref{algo:3}, respectively. The iteration is terminated when the numerical error is less than $10^{-9}$. It can be observed that these two method obtain the same U-eigenvalues; however, Algorithm \ref{algo:3} always takes less CPU time and iteration step than the neural network method. Of course, we should not say that Algorithm 3.3 is always better in all circumstances.
	
\begin{table}[!h]
	\centering
	\begin{tabular}{|c|c|c|c|c|c|}
		
		\hline
		$Tensors$&$\lambda_{\mathcal{A}}$& T1(sec) & T2(sec) & Iter1 & Iter2 \\
		\hline
		$2\times2\times 3$&0.7907&0.14& 0.03& 195& 21\\
		\hline
		$3\times3\times 3$&0.5798&0.47& 0.08& 352& 41\\
		\hline
		$3\times4 \times 5$&0.5729& 1.10&0.14& 403& 41\\
		\hline
		$2\times2 \times 2 \times 2$&0.6040&0.32&0.07& 276& 44\\
		\hline
		$3\times3 \times 5 \times 5$& 0.3569&57.86&5.89& 4407& 367\\
		\hline
		$2\times3 \times 4 \times 5\times 2$& 0.3298&47.58&4.98& 2510& 223\\
		\hline
	\end{tabular}
	\caption{Computational results for Example \ref{exam8}}\label{Tb8}
\end{table}

\end{example}

\section{Conclusion}

In this paper, we have proposed three different methods to compute U-eigenvalues of non-symmetric complex tensors and geometric measures of entanglement of non-symmetric pure states. The theory of symmetric embedding has been generalized from real tensors to complex tensors, and the relationship between U-eigenvalues of a non-symmetric complex tensor and US-eigenvalues of its symmetric embedding tensor has been established. Three algorithms have been given. Algorithm \ref{algo:1} computes the U-eigenvalues of a complex tensor by means of symmetric embedding, Algorithm \ref{algo:2} computes directly the U-eigenvalues of a complex tensor, and Algorithm \ref{algo:3} is a tensor version of the Gauss-Seidel method.  The convergence of Algorithms \ref{algo:1} and \ref{algo:2} has been proved.  Numerical examples are used to demonstrate Algorithms \ref{algo:1}-\ref{algo:3}, and it is observed that Algorithm \ref{algo:3} is more computationally efficient than the other two algorithms. The convergence analysis of Algorithm \ref{algo:3} is  our future research.

%%%%%%%%%%%%%%%%%%%%%%%%%%%%%%%%%%%%%%%%%%%%%%%%%%%%%%%%%

\begin{acknowledgements}
The authors would like to thank the editors and two anonymous referees for their valuable
suggestions, which helped us to improve this manuscript.
The first and the second authors' work is partially supported by the National Natural Science Foundation of China (No. 11871472), and the third author's work is partially
supported by the Hong Kong Research Grant Council under grants 15206915 and 15208418.
\end{acknowledgements}

% BibTeX users please use one of
%\bibliographystyle{spbasic}      % basic style, author-year citations
%\bibliographystyle{spmpsci}      % mathematics and physical sciences
%\bibliographystyle{spphys}       % APS-like style for physics
%\bibliography{}   % name your BibTeX data base

% Non-BibTeX users please use

\end{document}